\documentclass{amsart}
\usepackage{amsmath}
\usepackage{amssymb}
\usepackage{amsthm}
\usepackage{enumerate}
\usepackage[dvips]{graphicx}
\usepackage{float, pict2e, bxeepic}

\usepackage[usenames,dvipsnames]{xcolor}

\usepackage{capt-of}
\usepackage{tikz}

\newtheorem{theorem}{Theorem}

\newtheorem{lemma}{Lemma}
\newtheorem{rem}{Remark}

\begin{document}
\title{Chow's theorem for Hilbert Grassmannians as a Wigner-type theorem}
\author{Mark Pankov, Adam Tyc}

\address{Mark Pankov: Faculty of Mathematics and Computer Science, 
University of Warmia and Mazury, S{\l}oneczna 54, 10-710 Olsztyn, Poland}
\email{pankov@matman.uwm.edu.pl}
\address{Adam Tyc: Faculty of Mathematics and Computer Science, 
University of Warmia and Mazury, S{\l}oneczna 54, 10-710 Olsztyn, Poland}
\email{adam.tyc@matman.uwm.edu.pl}

\begin{abstract}
Let $H$ be an infinite-dimensional complex Hilbert space.
Denote by ${\mathcal G}_{\infty}(H)$ the Grassmannian formed by closed subspaces of $H$ whose dimension and codimension both are infinite.
We say that $X,Y\in {\mathcal G}_{\infty}(H)$ are {\it ortho-adjacent} if they are compatible and $X\cap Y$ is a hyperplane in both $X,Y$.
A subset ${\mathcal C}\subset {\mathcal G}_{\infty}(H)$ is called an $A$-{\it component}  if 
for any $X,Y\in {\mathcal C}$ the intersection $X\cap Y$ is of the same finite codimension in both $X,Y$ 
and ${\mathcal C}$ is maximal with respect to this property.
Let $f$ be a bijective transformation of ${\mathcal G}_{\infty}(H)$ preserving the ortho-adjacency relation in both directions.
We show that the restriction of $f$ to every $A$-component  is induced by a unitary or anti-unitary operator
or it is the composition of the orthocomplementary map and a map induced by a unitary or anti-unitary operator.
Note that the restrictions of $f$ to distinct $A$-components can be related to different operators.
\end{abstract}

\maketitle

\section{Introduction}
Two (not necessarily finite-dimensional) subspaces of a vector space are called {\it adjacent} if their intersection is a hyperplane in each of 
these subspaces \cite{BH,PankGCB,Pankov-book}.
Any two adjacent subspaces are of the same dimension and codimension.
Classic Chow's theorem \cite{Chow} characterizes semilinear automorphisms of vector spaces over division rings as 
adjacency preserving transformations of Grassmannians formed by finite-dimensional subspaces.
This result is a generalization of the Fundamental Theorem of Projective Geometry
characterizing semilinear automorphisms as automorphisms of projective spaces.
There are analogues of Chow's theorem for totally isotropic subspaces of sesquilinear forms, singular subspaces of quadratic forms
and various types of matrix spaces \cite{Die,PankGCB,Wan}.
An extension of Chow's theorem on Grassmannians consisting of subspaces with infinite dimension and codimension is an open problem.
We refer to \cite[Section 2.4]{Pankov-book} for an explanation why this problem is hard.

Chow's theorem was used to prove some Wigner-type theorems for Hilbert Grassmannians formed by finite-dimensional subspaces \cite{GS,Geher1,Pankov-book}.
These Grassmannians are identified with  the  conjugacy  classes  of  finite-rank projections;
the adjacency relation can be interpreted as follows: 
the ranges of projections are adjacent if their difference is an operator of rank $2$ (i.e. the smallest possible).

Classic Wigner’s theorem \cite{Wigner} characterizes unitary and anti-unitary operators as quantum symmetries;
see \cite{Chev} for a brief history and physical background.
By Gleason's theorem \cite{Gleason}, 
pure states of a quantum mechanical system can be identified with rank-one projections or, equivalently, rays of a complex Hilbert space $H$.
The non-bijective version of Wigner’s theorem states that every (not necessarily bijective) transformation of the pure state space preserving the transition probability, i.e. the angles between rays,
is induced by a linear or conjugate-linear isometry.
On the other hand, Uhlhorn's version of Wigner’s theorem \cite{Uhlhorn} says that every bijective transformation preserving the orthogonality of rays in both directions is induced by 
a unitary or anti-unitary operator if $\dim H\ge 3$; the assumption of bijectivity cannot be omitted if $H$ is infinite-dimensional.
Some recent generalizations of Wigner's theorem can be found in \cite{GeherWig,GM,PV,SemrlWig}.

In \cite{Molnar1,Molnar2}, the non-bijective version of Wigner's theorem is extended on the Grassmannian of $k$-dimensional subspaces of $H$ as follows: 
every (not necessarily bijective) transformation preserving the principal angles between subspaces is induced by  a linear or conjugate-linear isometry if $\dim H\ne 2k$;
in the case when $\dim H=2k$, such a transformation can be also the composition of the orthocomplementary map and a transformation induced by a linear conjugate-linear isometry.
Also, there is an extension of Uhlhorn's version. 
By \cite{Gyory,Semrl1}, every bijective transformation of the Grassmannian preserving the orthogonality relation in both directions is 
induced by a unitary or anti-unitary operator if $\dim H>2k$ (the orthogonality relation is not defined if $\dim H<2k$, for the case when $\dim H=2k$ the statement fails);
as above, the bijectivity assumption  cannot be dropped  if $H$ is infinite-dimensional.
Various examples of Wigner-type theorems for Hilbert Grassmannians can be found in \cite{Molnar-book,Pankov-book}.

In \cite{PPZ}, the adjacency relation on Hilbert Grassmannians of finite-dimensional subspaces is replaced by the ortho-adjacency
which provides a Chow-type theorem characterizing unitary and anti-unitary operators. 
Closed subspaces of a Hilbert space are called {\it ortho-adjacent} if they are adjacent and compatible. 
The compatibility of subspaces is equivalent to the commutativity of the corresponding projections.
Note that the ortho-adjacency relation for anisotropic subspaces of  sesquilinear forms is considered in \cite{Havlicek,PZ}.

Now, we assume that the Hilbert space $H$ is infinite-dimensional and consider the Grassmannian ${\mathcal G}_{\infty}(H)$ formed by closed subspaces of $H$ whose dimension and codimension both are infinite.
The Grassmannian is partially ordered by the inclusion relation and  every automorphism of this poset is induced by an invertible bounded linear or conjugate-linear operator
\cite[Theorem 3.17]{Pankov-book}.
This is closely related to \cite[Theorem 1.2]{Semrl1} which states that every bijective transformation of  ${\mathcal G}_{\infty}(H)$ preserving the orthogonality in both directions
is induced by a unitary or anti-unitary operator.
The orthocomplementary transformation of  ${\mathcal G}_{\infty}(H)$ is adjacency and ortho-adjacency preserving; however, it reverses inclusions and does not preserve the orthogonality.
Furthermore, there are bijective transformations of  ${\mathcal G}_{\infty}(H)$ preserving the adjacency and ortho-adjacency in both directions and 
non-obtainable by unitary or anti-unitary operators and the possible composition with  the orthocomplementary map. 
This is due to the existence of elements of  ${\mathcal G}_{\infty}(H)$ whose intersection is of infinite codimension in each of them; 
such elements cannot be connected by a finite sequence, where consecutive elements are adjacent or ortho-adjacent.

In the present paper, we  show that every  bijective transformation of  ${\mathcal G}_{\infty}(H)$ preserving the ortho-adjacency in both directions {\it locally} is induced by a unitary or anti-unitary operator 
with the possible composition with  the orthocomplementary map. 
At the end (Section 5), we explain why a similar statement is not proved for the adjacency relation and invertible bounded linear or conjugate-linear operators.

\section{Main result}
Let $H$ be a complex Hilbert space of dimension not less than $3$. 
For every positive integer $k<\dim H$ we denote by ${\mathcal G}_k(H)$ and ${\mathcal G}^k(H)$
the Grassmannians formed by $k$-dimensional subspaces of $H$ and closed subspaces of $H$ whose codimension is $k$, respectively.
Note that ${\mathcal G}^k(H)={\mathcal G}_{n-k}(H)$ if $\dim H=n$ is finite.
In the case when $H$ is infinite-dimensional,  we write ${\mathcal G}_{\infty}(H)$ for the Grassmannian of closed subspaces of $H$ whose dimension and codimension both are infinite. 

Let ${\mathcal G}$ be one of the Grassmannians ${\mathcal G}_k(H),{\mathcal G}^k(H)$ or ${\mathcal G}_{\infty}(H)$.
Following \cite{BH,PankGCB,Pankov-book}, we say that
$X,Y\in {\mathcal G}$ are {\it adjacent} if $X\cap Y$ is a hyperplane in both $X,Y$.
In the case when ${\mathcal G}$ is ${\mathcal G}_1(H)$ or ${\mathcal G}^1(H)$, any two distinct elements of ${\mathcal G}$ are adjacent. 
Recall that two closed subspaces of $H$ are {\it compatible} if there is an orthonormal basis of $H$ such that each of these subspaces is spanned by a subset of this basis.
Subspaces $X,Y\in {\mathcal G}$ are called {\it ortho-adjacent}  if they are adjacent and compatible. 
If ${\mathcal G}$ is ${\mathcal G}_k(H)$ or ${\mathcal G}^k(H)$, 
then any $X,Y\in {\mathcal G}$ can be connected by a finite sequence of elements of ${\mathcal G}$ such that any two consecutive elements are adjacent.
If ${\mathcal G}={\mathcal G}_\infty(H)$, then the same holds only in the case when
$X\cap Y$ has the same finite codimension in both $X,Y$; see \cite[Section 2.3]{Pankov-book}.
For any adjacent $X,Y\in {\mathcal G}$ there is $Z\in {\mathcal G}$ ortho-adjacent to both $X,Y$.
Therefore, if two elements of ${\mathcal G}$ are connected via the adjacency, then they are also connected via the ortho-adjacency.
A subset ${\mathcal C}\subset {\mathcal G}_{\infty}(H)$ is said to be an $A$-{\it component}  if 
for any $X,Y\in {\mathcal C}$ the intersection $X\cap Y$ is of the same finite codimension in both $X,Y$ 
(in other words, $X,Y$ are connected via the adjacency) and ${\mathcal C}$ is maximal with respect to this property.

Every unitary or anti-unitary operator on $H$ induces bijective transformations of ${\mathcal G}_k(H),{\mathcal G}^k(H)$ and ${\mathcal G}_{\infty}(H)$
which preserve the adjacency and ortho-adjacency  relations in both directions. 
The orthocomplementary map $X\to X^{\perp}$ induces bijections of ${\mathcal G}_k(H)$ and ${\mathcal G}^k(H)$ 
to ${\mathcal G}^k(H)$  and ${\mathcal G}_k(H)$ (respectively)  and a bijective transformation of ${\mathcal G}_{\infty}(H)$
preserving the adjacency and ortho-adjacency  relations in both directions.

Two $1$-dimensional subspaces of $H$ are ortho-adjacent if and only if they are orthogonal. 
By Uhlhorn's version of Wigner's theorem \cite{Uhlhorn}, every bijective transformation of ${\mathcal G}_1(H)$  preserving the orthogonality relation in both directions
is induced by a unitary or anti-unitary operator.
The main result of \cite{PPZ} states that the same holds for  every bijective transformation of ${\mathcal G}_k(H)$  preserving the ortho-adjacency relation in both directions if 
$\dim H\ne 2k$;
in the case when $\dim H=2k\ge 6$, every such transformation is induced by a unitary or anti-unitary operator or it is the composition of the orthocomplementary map and 
a transformation induced by a unitary or anti-unitary operator.
If $\dim H=2k=4$, then the latter statement fails and a description of ortho-adjacency preserving transformations is an open problem, see \cite{Havlicek, PPZ} for the details.

\begin{rem}{\rm
If $H$ is infinite-dimensional and $f$ is a bijective transformation of ${\mathcal G}^k(H)$ preserving the ortho-adjacency in both directions,
then $X\to f(X^{\perp})^{\perp}$ is a bijective transformation of ${\mathcal G}_k(H)$ which also preserves the ortho-adjacency in both directions.
If the latter transformation is induced by a unitary or anti-unitary operator $U$, then $f$ is also induced by $U$. 
}\end{rem}

The following example shows that  the above statement fails for bijective transformations of ${\mathcal G}_{\infty}(H)$  preserving the ortho-adjacency relation in both directions.
Let $U$ be a non-identity unitary operator on $H$ which preserves a certain $A$-component ${\mathcal C}\subset {\mathcal G}_{\infty}(H)$.
Consider the bijective transformation $f$ of ${\mathcal G}_{\infty}(H)$ defined as follows: $f(X)=U(X)$ if $X\in {\mathcal C}$ and $f(X)=X$ if $X\not\in {\mathcal C}$.
It is clear that $f$ is ortho-adjacency preserving in both directions.

\begin{theorem}\label{theorem-main}
Let $f$ be a bijective transformation of ${\mathcal G}_{\infty}(H)$ which preserves the ortho-adjacency relation in both directions.
Then the restriction of $f$ to every $A$-component of ${\mathcal G}_{\infty}(H)$ is induced by a unitary or anti-unitary operator or it is 
the composition of the orthocomplementary map and a map induced by a unitary or anti-unitary operator.
\end{theorem}

Note that the restrictions of $f$ to distinct $A$-components can be related to  different operators.

\section{A characterization of adjacency in terms of ortho-adjacency}
If $X,Y\in{\mathcal G}_{\infty}(H)$ are not ortho-adjacent and there is $Z\in \mathcal{G}_{\infty}(H)$ ortho-adjacent to both $X,Y$, 
then one of the following possibilities is realized:
\begin{enumerate}
\item[$\bullet$] $X,Y$ are adjacent;
\item[$\bullet$] $X\cap Y$ is of codimension $2$ in both $X, Y$. 
\end{enumerate}
Indeed,  $X\cap Z$ and $Y\cap Z$ are hyperplanes of $Z$ and $X,Y$ are adjacent if these hyperplanes coincide;
the same holds if the hyperplanes are distinct and their intersection is a proper subspace of $X\cap Y$;
we obtain the second possibility only when the hyperplanes are distinct and their intersection coincides with $X\cap Y$.

\begin{lemma}\label{lemma-ch1}
Let $X, Y$ be compatible elements of $\mathcal{G}_{\infty}(H)$ whose intersection is of codimension $2$ in both $X, Y$. 
Then there is  $Z\in\mathcal{G}_{\infty}(H)$ ortho-adjacent to both $X,Y$.
For every such $Z$ there are precisely two elements of $\mathcal{G}_{\infty}(H)$ ortho-adjacent to each of $X,Y,Z$.
\end{lemma}

\begin{proof}
The orthogonal sum of $X\cap Y$, a $1$-dimensional subspace of $X\cap (X\cap Y)^{\perp}$ and a $1$-dimensional subspace of $Y\cap (X\cap Y)^{\perp}$ 
is an element of $\mathcal{G}_{\infty}(H)$ ortho-adjacent to both $X,Y$.

Let $Z$ be an element of $\mathcal{G}_{\infty}(H)$ ortho-adjacent to both $X,Y$.
The subspaces $X,Y,Z$ are mutually compatible and there is an orthonormal basis $B$ of $H$
such that each of these subspaces is spanned by a subset of $B$.
Then $X\cap Y$ and the $2$-dimensional subspaces 
$$X'=X\cap (X\cap Y)^{\perp},\;\;Y'=Y\cap (X\cap Y)^{\perp}$$
are also spanned by subsets of $B$. 
Furthermore, $X\cap Y$ is contained in $Z$ and 
$$P=Z\cap X',\;\;Q=Z\cap Y'$$
are $1$-dimensional.
Let $P'$ and $Q'$ be the $1$-dimensional subspaces 
which are the  orthogonal complements of $P$ in $X'$ and $Q$ in $Y'$, respectively.
Since each of the subspaces $P,P',Q,Q'$ contains a vector from the basis $B$, 
the subspaces
$$Z_1= P' \oplus (X\cap Y)\oplus Q\;\mbox{ and }\; Z_2=P\oplus (X\cap Y) \oplus Q'$$
are spanned by subsets of $B$ and ortho-adjacent to each of $X,Y,Z$. 

Consider any $Z'\in \mathcal{G}_{\infty}(H)$ ortho-adjacent to each of $X,Y,Z$.
The subspaces $X,Y,Z,Z'$ are mutually compatible and there is an orthonormal basis $B'$ such that all these subspaces are spanned by subsets of $B'$. 
Then $X\cap Y, X',Y'$ are also spanned by subsets of $B'$
and each of the $1$-dimensional subspaces $P,P',Q,Q'$ contains a vector from $B'$.
The subspace $Z'$ contains $X\cap Y$,
the subspaces $X'\cap Z'$, $Y'\cap Z'$ are $1$-dimensional and 
$$Z'= (X'\cap Z')\oplus (X\cap Y)\oplus (Y'\cap Z')$$
($Z'$ is ortho-adjacent to $X$ and $Y$).
Recall that 
$$Z=P\oplus (X\cap Y)\oplus Q.$$
Since  $X'\cap Z'$ is a $1$-dimensional subspace of $X'$ containing a vector from $B'$,
it coincides with $P$ or $P'$.
Similarly, $Y'\cap Z'$ coincides with $Q$ or $Q'$.
Then $Z$ and $Z'$ are ortho-adjacent only when $Z'$ is $Z_1$ or $Z_2$.
\end{proof}

\begin{lemma}\label{lemma-ch2}
Let $X, Y$ be elements of $\mathcal{G}_{\infty}(H)$ whose intersection is of codimension $2$ in both $X, Y$. 
If there are ortho-adjacent $Z,Z'\in\mathcal{G}_{\infty}(H)$ such that each of them 
is ortho-adjacent to both $X,Y$, then $X, Y$ are compatible.
\end{lemma}

\begin{proof}
We assert that $Z\cap Z'$ is contained in $X$ or $Y$.
If $Z\cap Z'$ is not contained in $X$, then $Z\cap X$ and $Z'\cap X$ are distinct hyperplanes of $X$ and their sum is $X$. 
Similarly, if $Z\cap Z' \not\subset Y$, then  $Z\cap Y$ and $Z'\cap Y$ are distinct hyperplanes of $Y$ whose sum is $Y$.
Then $Z+Z'$ contains both $X,Y$. 
Since $Z\cap X$ and $Z'\cap X$ are subspaces of codimension $2$ in $Z+Z'$, the subspace $X$ is a hyperplane of $Z+Z'$. 
For the same reason, $Y$ is a hyperplane of $Z+Z'$. 
Then $X,Y$ are adjacent which is impossible. 

Without loss of generality, we can assume that $Z\cap Z'$ is contained in $X$.
Then $Z\cap Z'$ is a hyperplane of $X$. 
The orthogonal complement of $X\cap Y$ in $X+Y$ is $4$-dimensional and we denote this subspace by $M$. 
Then
$$X'=X\cap M\;\mbox{ and }\; Y'=Y\cap M$$
are $2$-dimensional subspaces whose intersection is $0$. 
Observe that each of $Z,Z'$ contains $X\cap Y$ and is contained in $X+Y$
which implies that
$$S=Z\cap M\;\mbox{ and }\; S'=Z'\cap M$$
are distinct $2$-dimensional subspaces. 
Since $Z\cap Z'$ is a hyperplane of $X$ and $X\cap Y$ is contained in $Z\cap Z'$,
$$Z\cap Z'\cap M =S\cap S'$$
is a $1$-dimensional subspace of $X'$ which will be denoted by $P$.
The subspaces $$Z=(X\cap Y)\oplus S,\;\;Z'=(X\cap Y)\oplus S'$$
are ortho-adjacent to $Y=(X\cap Y)\oplus Y'$ and, consequently,
$S$ and $S'$ intersect $Y'$ in certain $1$-dimensional subspaces $P_1$ and $P_2$,
respectively.
The subspaces $P_1, P_2$ are distinct (since $S=P+P_1$ and $S'=P+P_2$ are distinct).

\begin{center}
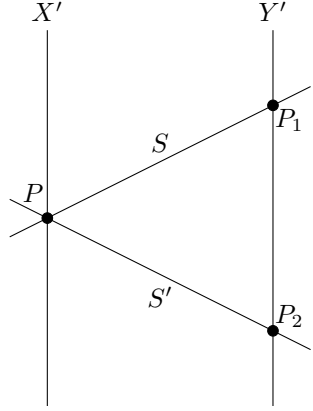

\begin{tikzpicture}

\draw[color=black] (0,0)--(0,5);
\draw[color=black] (3,0)--(3,5);

\draw[color=black] (-0.5,2.75)--(3.5,0.75);
\draw[color=black] (-0.5,2.25)--(3.5,4.25);

\draw[fill=black] (0,2.5) circle (2pt); 

\draw[fill=black] (3,1) circle (2pt); 
\draw[fill=black] (3,4) circle (2pt);

\node at (0,5.25) {$X'$};
\node at (3,5.25) {$Y'$};

\node at (1.5,3.5) {$S$};
\node at (1.5,1.45) {$S'$};

\node at (3.22,3.8) {$P_1$};
\node at (3.22,1.2) {$P_2$};

\node at (-0.187,2.83) {$P$};

\end{tikzpicture}
\captionof{figure}{The intersections of $X', Y'$ with $S$ and $S'$}
\end{center}

Show that $S$ and $S'$ are compatible to both $X',Y'$. 
The following statement is well-known
\cite[Lemma 1.14]{Pankov-book}: 
if a closed subspace $A\subset H$ is compatible to closed subspaces $B,C\subset H$, 
then $A$ is compatible to $B\cap C$, $\overline{B+C}$ and $B^{\perp}$.
This implies that $Z$ is compatible to $X\cap Y$ and $X+Y$ 
(since $Z$ is compatible to $X,Y$).
Then $Z$ is compatible to 
$$M=(X+Y)\cap (X\cap Y)^{\perp}$$
and, consequently, to $X'=X\cap M$. So, $X'$ is compatible to $Z$ and $M$ 
($X'$ is contained in $M$) which means that $X'$ is compatible to $S=Z\cap M$. 
For the same reason, $X'$ is compatible to $S'$ and $Y'$ is compatible to both $S,S'$.

To complete our proof we need to show that $X'$ and $Y'$ are orthogonal. 

Recall that $P=S\cap S'$ is a $1$-dimensional subspace of $X'$.
Let $Q$ be the unique $1$-dimensional subspace of $X'$ orthogonal to $P$.
Since $X'$ is compatible to $S$ and $S'$, we obtain that $Q$ is orthogonal to $S$ and $S'$.
Therefore, $Q$ is orthogonal to $S+S'$ and, consequently, to  $Y'=P_1+P_2\subset S+S'$.

Show that $P$ is orthogonal to $Y'$. 
Let $Q_i$, $i=1,2$ be the $1$-dimensional subspace of $Y'$ orthogonal to $P_i$.
Then $Q_1\ne Q_2$, since $P_1\ne P_2$.
Furthermore, $Q_1$ is orthogonal to $S$ (since $S$ and $Y'$ are compatible) and, similarly, 
$Q_2$ is orthogonal to $S'$. 
This means that $Q_1,Q_2$ both are orthogonal to $P=S\cap S'$
and $Y'=Q_1+Q_2$ is orthogonal to $P$.

So, $X'=P+Q$ is orthogonal to $Y'$. This implies that $X$ and $Y$ are compatible.
\end{proof}


\begin{lemma}\label{lemma-main}
For distinct $X,Y\in {\mathcal G}_{\infty}(H)$ the following conditions are equivalent:
\begin{enumerate}
\item[{\rm (1)}] $X,Y$ are adjacent {\rm(}not necessarily ortho-adjacent{\rm)};
\item[{\rm (2)}] there are infinitely many $Z\in {\mathcal G}_{\infty}(H)$ ortho-adjacent to both $X,Y$ such that there are infinitely many $Z'\in {\mathcal G}_{\infty}(H)$ ortho-adjacent to $X,Y,Z$. 
\end{enumerate}
\end{lemma}

\begin{proof}
$(1)\Rightarrow(2)$.
Since $(X+Y)^{\perp}$ is infinite-dimensional, for every $1$-dimensional $P\subset (X+Y)^{\perp}$ there are infinitely many 
$1$-dimensional subspaces $Q\subset (X+Y)^{\perp}$ orthogonal to $P$.
For any such $P$ and $Q$ the ortho-adjacent subspaces
$$
P+(X\cap Y), Q+(X\cap Y)
$$
are ortho-adjacent to each of $X,Y$.

$(2)\Rightarrow(1)$.
Suppose that $X,Y$ are not adjacent. Then $X\cap Y$ is of codimension $2$ in both $X,Y$.
Since there is a pair of ortho-adjacent  elements of ${\mathcal G}_{\infty}(H)$ which are ortho-adjacent to both $X,Y$,
Lemma \ref{lemma-ch2} implies that $X$ and $Y$ are compatible.
Then, by Lemma \ref{lemma-ch1}, for every  $Z\in {\mathcal G}_{\infty}(H)$ ortho-adjacent to both $X,Y$
there are precisely two elements of ${\mathcal G}_{\infty}(H)$ ortho-adjacent to $X,Y,Z$
which contradicts our assumption. 
Therefore, $X,Y$ are adjacent. 
\end{proof}

\section{Proof of Theorem \ref{theorem-main}}
Let $f$ be a bijective transformation of ${\mathcal G}_{\infty}(H)$ preserving the ortho-adjacency in both directions.
Lemma \ref{lemma-main} shows that $f$ also is adjacency preserving in both directions. 
Then $f$ preservers the family of subsets maximal with respect to the property that any two distinct elements are adjacent. 
By \cite[Proposition 2.14]{Pankov-book}, every such subset is of one of the following types: 
\begin{enumerate}
\item[$\bullet$] the star ${\mathcal S}(X)$, $X\in {\mathcal G}_{\infty}(H)$ which consists of all closed subspaces of $H$ containing $X$ as  a hyperplane;
\item[$\bullet$] the top ${\mathcal G}^{1}(X)$, $X\in {\mathcal G}_{\infty}(H)$.
\end{enumerate}
Let ${\mathcal C}$ be an $A$-component of ${\mathcal G}_{\infty}(H)$.
Denote by ${\mathcal C}_{+1}$ and ${\mathcal C}_{-1}$
the sets of all $X\in {\mathcal G}_{\infty}(H)$ such that $X$  contains a certain element of ${\mathcal C}$ as a hyperplane
or $X$ is a hyperplane in a certain element of ${\mathcal C}$, respectively.
Then $X\in {\mathcal G}_{\infty}(H)$ belongs to ${\mathcal C}_{+1}$ if and only if the top ${\mathcal G}^{1}(X)$ is contained in ${\mathcal C}$;
similarly, $X$  belongs to ${\mathcal C}_{-1}$ if and only if the star ${\mathcal S}(X)$ is contained in ${\mathcal C}$.
It is easy to check  that ${\mathcal C}_{+1}$ and ${\mathcal C}_{-1}$ are $A$-components of ${\mathcal G}_{\infty}(H)$.
By \cite[Theorem 2.19]{Pankov-book},  one of the following possibilities is realized:
\begin{enumerate}
\item[(1)] $f$ sends every top ${\mathcal G}^1(X)$, $X\in {\mathcal C}_{+1}$ to a top and every star ${\mathcal S}(Y)$, $Y\in {\mathcal C}_{-1}$ to a star;
\item[(2)] $f$ transfers all tops ${\mathcal G}^1(X)$, $X\in {\mathcal C}_{+1}$ to stars and all stars ${\mathcal S}(Y)$, $Y\in {\mathcal C}_{-1}$ to tops.
\end{enumerate}
In the case (2), the composition of the orthocomplementary map and $f|_{\mathcal C}$ sends 
tops to tops and stars to stars. Therefore, it is sufficient to show that $f|_{\mathcal C}$ is induced by a unitary or anti-unitary operator  in the case (1).

From this moment, we assume that  (1) is realized.

\begin{lemma}\label{lemma-p1}
Let $X\in {\mathcal C}_{+1}$. If
$f({\mathcal G}^1(X))={\mathcal G}^1(X')$,
then there is a unitary or anti-unitary operator $U_X:X\to X'$ such that
$$f(M)=U_X(M)$$
for every $M\in {\mathcal G}^1(X)$.
\end{lemma}

\begin{proof}
Consider the bijective map $g:{\mathcal G}_1(X)\to {\mathcal G}_1(X')$ defined as follows:
$$g(P)=f(P^{\perp}\cap X)^{\perp}\cap X'$$
for every $1$-dimensional subspace $P\subset X$.
Observe that $1$-dimensional subspaces $P,Q\subset X$ are orthogonal if and only if 
$$P^{\perp}\cap X,Q^{\perp}\cap X\in {\mathcal G}^1(X)$$ are ortho-adjacent;
the same holds for $1$-dimensional subspaces of $X'$.
This implies that $g$ is orthogonality preserving in both directions, since $f$ is ortho-adjacency preserving in both directions.
By Uhlhorn's theorem, $g$ is induced by a unitary or anti-unitary operator $U_X:X\to X'$.
Then 
$$f(M)=g(M^{\perp}\cap X)^{\perp}\cap X'=(U_X(M^{\perp}\cap X))^{\perp}\cap X'=U_X(M)$$
for every $M\in {\mathcal G}^1(X)$.
\end{proof}

\begin{lemma}\label{lemma-p2}
If $Z\in {\mathcal C}_{-1}$ is contained in $X\in{\mathcal C}_{+1}$, then
$$f({\mathcal S}(Z))={\mathcal S}(U_X(Z)).$$
\end{lemma}

\begin{proof}
We take distinct $M,N$ belonging to ${\mathcal S}(Z)\cap {\mathcal G}^1(X)$.
Then $Z=M\cap N$ and
$$f({\mathcal S}(Z))={\mathcal S}(Z'),$$
where $Z'=f(M)\cap f(N)$.
By Lemma \ref{lemma-p1},
$$U_X(Z)=U_X(M)\cap U_X(N)=f(M)\cap f(N)=Z'$$
which gives the claim.
\end{proof}

\begin{lemma}\label{lemma-p3}
If $X,Y\in {\mathcal C}_{+1}$, then 
$$U_X(P)=U_Y(P)$$
for every $1$-dimensional subspace $P\subset X\cap Y$.
\end{lemma}

\begin{proof}
Since ${\mathcal C}_{+1}$ is an $A$-component of ${\mathcal G}_{\infty}(H)$, 
there is a sequence 
$$X=X_0,X_1,\dots,X_{m}=Y$$
of elements of ${\mathcal C}_{+1}$, where $X_{i-1},X_i$ are adjacent for every $i\in \{1,\dots,m\}$;
furthermore, we can construct this sequence such that each $X_i$ contains $X\cap Y$.
For this reason, it is sufficient to consider the case when $X,Y$ are adjacent.

If $X$ and $Y$ are adjacent, then $M=X\cap Y$ belongs to ${\mathcal C}$. 
Let $P$ be a $1$-dimensional subspace of $M$. Then $Z=P^{\perp}\cap M$ belongs to ${\mathcal C}_{-1}$. Since $Z$ is contained in both $X,Y$, we have 
\begin{equation}\label{eq-p1}
U_X(Z)=U_Y(Z)
\end{equation}
by Lemma \ref{lemma-p2}.
Then
$$U_X(P)\oplus U_X(Z)=U_X(M)=f(M)=U_Y(M)=U_Y(P)\oplus U_Y(Z)$$
and \eqref{eq-p1} implies that $U_X(P)=U_Y(P)$.
\end{proof}

Consider the transformation $h$ of ${\mathcal G}_1(H)$ defined as follows: 
for a $1$-dimensional subspace $P\subset H$ we take any $X\in {\mathcal C}_{+1}$ containing $P$ and set $h(P)=U_X(P)$. 
By Lemma \ref{lemma-p3}, $h$ is well-defined. 
Since for any two $1$-dimensional subspaces of $H$ there is an element of ${\mathcal C}_{+1}$ containing them,
$h$ is a bijection preserving the orthogonality relation in both directions. Then $h$ is induced by a unitary or anti-unitary operator $U$ on $H$ (Uhlhorn's theorem).
The restriction of $U$ to every $X\in {\mathcal C}_{+}$ is a scalar multiple of $U_X$. This implies that $f|_{\mathcal C}$ is induced by $U$.

\section{Final remarks}
{\bf 1}. By  \cite[Theorem 3.17]{Pankov-book},
every automorphism of the poset $({\mathcal G}_{\infty}(H),\subset)$ is induced by an invertible bounded linear or conjugate-linear operator.
We explain why the same statement  is not proved for the restrictions of adjacency preserving transformations to $A$-components of  ${\mathcal G}_{\infty}(H)$.

Let $f$ be a bijective transformation of ${\mathcal G}_{\infty}(H)$ preserving the adjacency relation in both directions
and let ${\mathcal C}$ be an $A$-component of ${\mathcal G}_{\infty}(H)$.
Without loss of generality, we can assume that $f({\mathcal C})={\mathcal C}$. 
Denote by ${\mathcal C}_{\pm}$ 
the set of all $X\in {\mathcal G}_{\infty}(H)$ such that $X$ is a subspace of finite codimension in a certain element of ${\mathcal C}$ or 
$X$ contains a certain  element of ${\mathcal C}$ as a subspace of finite codimension.
Then $f$ can be uniquely extended to an automorphism of the poset $({\mathcal C}_{\pm},\subset)$ \cite[Theorem 2.19]{Pankov-book}.
This extension is denoted by the same symbol $f$.

Suppose that $f(X)=X$ for a certain $X\in {\mathcal C}$ and 
consider the lattice of finite-dimensional subspaces of $X$.
The map sending every finite-dimensional subspace $M\subset X$ to $f(M^{\perp}\cap X)^{\perp}\cap X$ is a lattice automorphism. 
By the Fundamental Theorem of Projective Geometry, it is induced by a semilinear bijection $A:X\to X$, i.e. 
$A$ is additive and there is an automorphism $\sigma$ of the field of complex numbers such that $A(ax)=\sigma(a)A(x)$
for every $x\in X$ and $a\in {\mathbb C}$. 
It must  be pointed out that $A$ is not necessarily bounded; furthermore, the automorphism $\sigma$ is not necessarily the identity or the conjugate map. 
Then
$$f(Y)=(A(Y^{\perp}\cap X))^{\perp}\cap X$$
for every $Y\in {\mathcal C}_{\pm}$ contained in $X$.

If $A$ is an invertible bounded  linear or conjugate-linear operator,
then the restriction of $f$ to the set of all $Y\in {\mathcal C}_{\pm}$ contained in $X$ is induced by $(A^{-1})^*$ \cite[Proposition 3.7]{Pankov-book}.

Conversely, if this restriction is induced by a semilinear bijection $B:X\to X$,
then $B$ sends closed hyperplanes of $X$ to closed hyperplanes which guarantees that
$B$ is a bounded  linear or conjugate-linear operator (see \cite{KM} or \cite[Lemma 3.12]{Pankov-book}).
By \cite[Proposition 3.7]{Pankov-book}, $A$ is a scalar multiple of $(B^{-1})^{*}$ which is impossible if $A$ is unbounded. 

So, the restriction cannot be induced by a semilinear bijection if $A$ is unbounded.
To assert  that $A$  is bounded we need to extend $f$ on ${\mathcal G}_{\infty}(X)$.
Only in this case, we can apply arguments used to prove \cite[Theorem 3.17]{Pankov-book}.

{\bf 2}.  We say that two elements of ${\mathcal G}_{\infty}(H)$ are {\it rank-one incident} if one of them is a hyperplane of the other.
Then $X,Y\in {\mathcal G}_{\infty}(H)$ are connected via this relation if and only if $X\cap Y$ is of finite codimension in both $X,Y$
(the codimensions are not necessarily equal). The corresponding connected components are all ${\mathcal C}_{\pm}$, where 
${\mathcal C}$ is an $A$-component. By the above reasons, we cannot describe
the restriction of a bijective transformation of ${\mathcal G}_{\infty}(H)$ preserving the rank-one incidence (in both directions) to ${\mathcal C}_{\pm}$.


\begin{thebibliography}{999}

\bibitem{BH} 
A. Blunck, H. Havlicek, {\it On bijections that preserve complementarity of subspaces}, Discrete Math. 301(2005), 46-56. 

\bibitem{Chev}
G. Chevalier, {\it Wigner’s theorem and its generalizations}, Handbook of Quantum
Logic and Quantum Structures, 429--475, Elsevier, 2007.

\bibitem{Chow} 
W.L. Chow, {\it On the geometry of algebraic homogeneous spaces}, Ann. of Math. 50(1949), 32--67.

\bibitem{Die}
J. Dieudonn\'e, {\it La g\'eom\'etrie des groupes classiques} (2nd edition), Springer, 1963.


\bibitem{GS}
G.P. Geh\'er, P. \v{S}emrl, 
{\it Isometries of Grassmann spaces}, 
J. Funct. Anal. 270 (2016), 1585--1601.

\bibitem{Geher1}
G.P. Geh\'er, {\it Wigner's theorem on Grassmann spaces}, 
J. Funct. Anal. 273(2017), 2994--3001.

\bibitem{GeherWig} 
G.P. Geh\'er, {\it Symmetries of projective spaces and spheres}. 
Int. Math. Res. Not. IMRN 2020, no. 7, 2205--2240.

\bibitem{GM}
G.P. Geh\'er, M. Mori,
{\it The structure of maps on the space of all quantum pure states that preserve a fixed quantum angle},
Int. Math. Res. Not. 2022, no. 16, 12003--12029.



\bibitem{Gleason}
A.M. Gleason, 
{\it Measures on the closed subspaces of a Hilbert space}, 
Indiana Univ. Math. J. 6(1957), 885--893.

\bibitem{Gyory}
M. Gy\H{o}ry,
{\it Transformations on the set of all $n$-dimensional subspaces of a Hilbert space preserving orthogonality}, 
Publ. Math. Debrecen 65 (2004), 233--242.

\bibitem{Havlicek}
H. Havlicek, 
{\it On Pl\"{u}cker transformations of generalized elliptic spaces}, Rend. Mat. Appl. (7) 15 (1995), 39--56.


\bibitem{KM}
S. Kakutani, G. W. Mackey, {\it Ring and lattice characterizations of complex Hilbert space}, Bull. Amer. Math. Soc. 52(1946), 727--733.


\bibitem{Molnar1}
L. Moln\'ar, {\it Transformations on the set of all $n$-dimensional subspaces of a Hilbert space preserving principal angles}, 
Comm. Math. Phys. 217 (2001), 409--421.

\bibitem{Molnar-book}
L. Moln\'ar, {\it Selected Preserver Problems on Algebraic Structures of Linear Operators and on Function Spaces},
Lecture Notes in Mathematics 1895, Springer, 2007.

\bibitem{Molnar2}
L. Moln\'ar, 
{\it Maps on the n-dimensional subspaces of a Hilbert space preserving principal angles}, 
Proc. Amer. Math. Soc. 136 (2008), 3205--3209.

\bibitem{PankGCB}
M. Pankov, {\it Grassmannians of classical buildings}, World Scientific, 2010. 

\bibitem{Pankov-book}
M. Pankov, {\it Wigner-type theorems for Hilbert Grassmannians}, 
London Mathematical Society Lecture Note Series 460, Cambridge University Press, 2020. 

\bibitem{PV}
M. Pankov, T. Vetterlein,
{\it A geometric approach to Wigner-type theorems}, Bull. London Math. Soc. 53 (2021) 1653--1662. 

\bibitem{PPZ}
M. Pankov, K. Petelczyc, M. \.Zynel, 
{\it Automorphisms and some geodesic properties of ortho-Grassmann graphs},
Electron. J. Comb. 28(2021), P4.49. 

\bibitem{PZ}
K. Pra\.zmowski, M. \.Zynel,
{\it Orthogonality of subspaces in metric-projective geometry},
Adv. Geom. 11(2011), 103--116.


\bibitem{Semrl1}
P. \v{S}emrl, {\it Orthogonality preserving transformations on the set of 
$n$-dimensional subspaces of a Hilbert space}, Illinois J. Math. 48 (2004), 567--573.

\bibitem{SemrlWig}
P. \v{S}emrl, {\it Wigner symmetries and Gleason's theorem},
J. Phys. A: Math. Theor. 54 (2021), 315301.


\bibitem{Uhlhorn}
U. Uhlhorn, {\it Representation of symmetry transformations in quantum mechanics}, 
Ark. Fysik 23 (1963), 307-340.


\bibitem{Wan}
Z. Wan, {\it Geometry of Matrices}, World Scientific, 1996.

\bibitem{Wigner}
E.P. Wigner, 
{\it Gruppentheorie und ihre Anwendung auf die Quantenmechanik der Atomspektren}, Fredrik Vieweg und Sohn, 1931 (English translation Group Theory and its Applications to the Quantum Mechanics of Atomic Spectra, Academic Press, 1959).

\end{thebibliography}
\end{document}